\documentclass[conference]{IEEEtran}
\usepackage{multirow}
\usepackage{booktabs}
\usepackage{amsfonts}
\usepackage{placeins}
\usepackage{array}
\usepackage{amsmath,cases,amssymb}
\usepackage{amsmath}
\usepackage{mathtools}
\usepackage{graphicx}
\usepackage{subcaption}
\usepackage{cite}
\usepackage{epsfig}
\usepackage{todonotes}
\usepackage{url}
\usepackage{algorithm}
\usepackage{algorithmic}
\usepackage{epstopdf}
\usepackage{balance}
\usepackage{color}
\DeclareGraphicsExtensions{.eps,.pdf}

\usepackage{scalerel}
\usepackage{multicol}
\usepackage{amssymb}% http://ctan.org/pkg/amssymb
\usepackage{pifont}% http://ctan.org/pkg/pifont
\usepackage{cases}

\usepackage{scalerel}
\usepackage{orcidlink}

\let\labelindent\relax
\usepackage{enumitem}

\usepackage{mathbbol}
\usepackage{cite}
\usepackage[T1]{fontenc} % optional
\usepackage{amsthm} % format theorem (bold title, italic content)
\usepackage{eqparbox}
\usepackage{textcomp}
\usepackage[eulergreek]{sansmath}
\usepackage{color}
\usepackage[cmintegrals]{newtxmath}

\usepackage[open,openlevel=1]{bookmark}
\usepackage{ctable}
\usepackage{threeparttablex}
\usepackage{makecell}

\newtheorem{theorem}{Theorem}

\newtheorem{proposition}{Proposition}

\allowdisplaybreaks

\newcommand{\bseq}{\begin{subequations}}
\newcommand{\eseq}{\end{subequations}}
\newcommand{\baln}{\begin{align}}
\newcommand{\ealn}{\end{align}}
\newcommand{\balnd}{\begin{aligned}}
\newcommand{\ealnd}{\end{aligned}}
\newcommand{\beq}{\begin{equation}}
\newcommand{\eeq}{\end{equation}}
\newcommand{\beqn}{\begin{eqnarray}}
\newcommand{\eeqn}{\end{eqnarray}}
\newcommand{\beqno}{\begin{eqnarray*}}
\newcommand{\eeqno}{\end{eqnarray*}}
\newcommand{\bma}{\begin{displaymath}}
\newcommand{\ema}{\end{displaymath}}
\newcommand{\bnu}{\begin{enumerate}}
\newcommand{\enu}{\end{enumerate}}
\newcommand{\bce}{\begin{center}}
\newcommand{\ece}{\end{center}}
\newcommand{\btb}{\begin{tabular}}
\newcommand{\etb}{\end{tabular}}
\newcommand{\bieq}{\begin{IEEEeqnarray}}
\newcommand{\eieq}{\end{IEEEeqnarray}}

\newcommand{\st}{{\mathrm{s.t.}}}
\newcommand{\subnum}{\IEEEyessubnumber}

\makeatletter
\newcommand{\linebreakand}{%
\end{@IEEEauthorhalign}
\hfill\mbox{}\par
\mbox{}\hfill\begin{@IEEEauthorhalign}
}
\makeatother

\begin{document}
\title{A Joint JSCC-Resource Allocation Framework for QoS-Aware Semantic Communication in LEO Satellite-based EO Missions}
% \vspace{-3mm}

\author{
    \IEEEauthorblockN{Hung Nguyen-Kha, Ti Ti Nguyen, Vu Nguyen Ha, Eva Lagunas, Symeon Chatzinotas, and Bjorn Ottersten}

    \IEEEauthorblockA{\textit{Interdisciplinary Centre for Security, Reliability and Trust (SnT), University of Luxembourg, Luxembourg} }

    % \IEEEauthorblockN{Hung Nguyen-Kha\orcidlink{0000-0002-5956-4279}, \IEEEmembership{Graduate Student Member,~IEEE,} Vu Nguyen Ha\orcidlink{0000-0003-1325-3480}, \IEEEmembership{Senior Member,~IEEE,}\\ Eva Lagunas\orcidlink{0000-0002-9936-7245}, \IEEEmembership{Senior Member,~IEEE,} Symeon Chatzinotas\orcidlink{0000-0001-5122-0001}, \IEEEmembership{Fellow,~IEEE},\\ and Joel Grotz\orcidlink{0000-0002-4095-4015}}, \IEEEmembership{Senior Member,~IEEE }

%\IEEEcompsocitemizethanks{ 
%% This work has been supported by the Luxembourg National Research Fund (FNR) under the project INSTRUCT (IPBG19/14016225/INSTRUCT).
%This research was funded in whole, or in part, by the Luxembourg National Research Fund (FNR) through the Project INtegrated Satellite – TeRrestrial Systems for Ubiquitous Beyond 5G CommunicaTions (INSTRUCT) under Grant IPBG19/14016225/INSTRUCT. 
%The preliminary result of this manuscript was presented in IEEE VTC Fall'24 \cite{Hung_VTC24}.
%}

%\IEEEcompsocitemizethanks{H. Nguyen-Kha, V. N. Ha, E. Lagunas, and S. Chatzinotas are with the Interdisciplinary Centre for Security, Reliability and Trust (SnT), University of Luxembourg, 1855 Luxembourg Ville, Luxembourg.  (e-mail: khahung.nguyen@uni.lu; vu-nguyen.ha@uni.lu; eva.lagunas@uni.lu; Symeon.Chatzinotas@uni.lu).}
%
%\IEEEcompsocitemizethanks{J. Grotz is with SES, Chateau de Betzdorf, Betzdorf 6815, Luxembourg (e-mail: Joel.Grotz@ses.com).}
\vspace{-10mm}
}

\maketitle

\begin{abstract} 
In Earth observation (EO) missions with Low Earth orbit (LEO) satellites, high-resolution image acquisition generates a massive data volume that poses a significant challenge for transmission under the limited satellite power budget, while LEO movement introduces dynamic systems. To enable efficient image transmission, this paper employs semantic communication (SemCom) with joint source-channel coding (JSCC), which focuses on transmitting meaningful information to reduce power consumption. Under a quality-of-service (QoS) requirement defined by image reconstruction quality, this work aims to minimize the total transmit power by jointly optimizing the JSCC encoder-decoder parameters and resource allocation. However, the implicit relationship among JSCC parameters, link quality, and image quality, coupled with the presence of mixed integer-continuous variables, makes the problem difficult to solve directly. To address this, a curve-fitting model is proposed to approximate the JSCC compression-SNR-quality relationship. Then, the joint compression ratio-resource allocation (JCRRA) algorithm is proposed to address the underlying problem.
%by reformulating it in a continuous domain, followed by a rounding step to recover discrete solutions and successive convex approximation (SCA) to address the non-convexity.
Numerical results demonstrate that the proposed method achieves substantial power savings compared to both greedy algorithms and conventional transmission paradigms.
\end{abstract}
%\vspace{-1mm}
%\begin{IEEEkeywords}
%%\vspace{-.3cm}
%LEO Constellation, Integrated Satellite-Terrestrial Networks, Seamless Connectivity, C-Band, 5G Automotive, Resource Allocation. 
%\end{IEEEkeywords}

%\vspace{-1cm}

%\newpage

\vspace{-2mm}
\section{Introduction}
\vspace{-2mm}
Recent developments in low Earth orbit (LEO) satellite (LSAT) constellations have enabled a broad range of services, including broadband connectivity, IoT communications, and Earth observation (EO) missions (e.g., ICEYE, HawkEye) \cite{SHAKUN20211743}. In EO applications, LSATs capture high-resolution imagery for downstream tasks such as environmental monitoring, disaster response, and ecosystem preservation \cite{LeyvaMayorga_TCOM23_EO}. However, the substantial volume of image data generated poses a critical bottleneck for transmission, given the limited bandwidth and stringent power constraints inherent in satellite links. As such, efficient image transmission frameworks and intelligent resource allocation (RA) strategies are essential to maximize system performance in LSAT-based EO platforms.

While recent efforts have investigated RA in LSAT systems \cite{VuHaGC2022, Zhu_IoTJ23_SatCom, Choi_IoTJ24_SatCom, Hung_WSA23, hung_ICCW_2023twotier, Hung_TWC24, Hung_VTC24, Hung_MeditCom24, Hung_TCOM25, hung_GLOBECOM25}, these studies primarily consider conventional bit-oriented transmission schemes. They often overlook the potential of semantic-aware techniques that exploit the underlying meaning and relevance of EO data, which could lead to more power-efficient and application-aligned communication paradigms.
In this context, semantic communication (SemCom) has emerged as a promising transmission paradigm, particularly well-suited for resource-constrained satellite systems. Unlike conventional communication frameworks transmitting raw or redundant information, SemCom focuses on conveying task-relevant and meaningful content \cite{Yang_ComST23_Semantic}. For many EO missions where objectives may not require full-resolution images but rather specific semantic cues - SemCom enables the transmission of essential content, thereby significantly reducing bandwidth and power consumption in LSAT systems.

%Recent work in \cite{Loc_TWC25_SemCom} explores the application of SemCom in satellite-based EO scenarios, employing separate learning-based source encoding, channel encoding, and denoising modules. However, the modular design of these components hinders the full exploitation of semantic redundancy and cross-layer optimization potential. 
To address this limitation, joint source-channel coding (JSCC) has been proposed as an effective SemCom technique tailored for data-type-aware transmission \cite{Gunduz_TCCN19_DJSCC, Ding_TWC24_SemCom, Hu_TCSVT2025_SemCom,Loc_TWC25_SemCom}. In particular, JSCC leverages deep learning architectures to directly map input data into channel symbols—bypassing the need for traditional source coding, channel coding, and modulation—and reconstructs semantic representations at the receiver end. This end-to-end approach enables more robust and efficient communication, especially under low SNR and dynamic channel conditions.
JSCC has garnered significant attention and has been extensively applied to image transmission within the SemCom framework. For instance, in \cite{Hu_TCSVT2025_SemCom}, the authors focused on optimizing computational resource utilization in SemCom-based image transmission. Meanwhile, the integration of SemCom into 5G communication systems has been explored in \cite{Ding_TWC24_SemCom, Zhang_JSAC25_SemCom}, where the coexistence of conventional bit-level and semantic transmission was addressed. Specifically, \cite{Ding_TWC24_SemCom} considers the joint support of SemCom and low-latency services, whereas \cite{Zhang_JSAC25_SemCom} investigates semantic-bit coexistence in a 5G multiple-input single-output (MISO) system. While these contributions provide valuable insights into the application of JSCC in terrestrial 5G contexts, they primarily target generic 5G-oriented SemCom scenarios.

To the best of our knowledge, the integration of JSCC in LSAT-based EO systems remains largely unexplored, highlighting a significant research gap in the literature.
In this paper, we study the LSAT EO systems employing the JSCC transmission paradigm. Under QoS requirement in terms of image reconstruction quality and limited LSAT's power budget, the objective is to minimize the total downlink transmit power. The main contributions are summarized as follows:
\begin{itemize}
    \item We study semantic satellite communication systems leveraging JSCC for practical EO applications and develop an optimization framework for power minimization through joint JSCC compression selection and RA under QoS and power constraints. %However, solving this problem is challenging due to the coupling between integer and continuous variables, as well as the lack of a closed-form expression linking image reconstruction quality to transmission parameters.
    \item To address the lack of a closed-form expression linking image reconstruction quality to transmission parameters, we propose an approximation model that captures the relationship among image quality, compression ratio, and signal-to-noise ratio (SNR) using a curve-fitting approach, thereby transforming the original problem into a more tractable form.
    \item To efficiently address the non-convex transformed problem, we design an iterative algorithm based on the successive convex approximation (SCA) method. In addition, two greedy algorithms--one using JSCC and the other using traditional JPEG compression--are considered for performance comparison.
    \item Numerical results, evaluated using a public dataset and realistic simulation parameters, demonstrate the superiority of the proposed algorithm over the benchmark schemes in terms of power minimization. Moreover, they highlight the effectiveness of JSCC for satellite EO systems compared with conventional transmission methods.
\end{itemize}

% ===========================================
\begin{figure}
    \centering
    \captionsetup{font=small}
    \includegraphics[width=0.75\columnwidth]{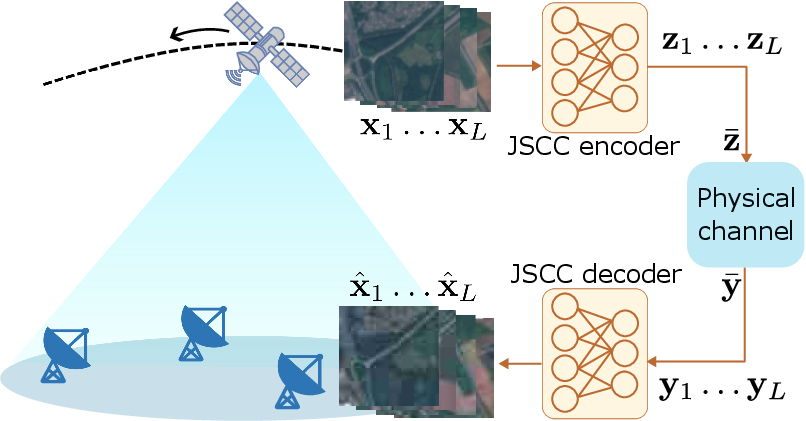}
    \caption{Satellite Earth-observation system model.}
    \label{fig:system_model}
    \vspace{-3mm}
\end{figure}
\vspace{-2mm}
\section{System Model and Problem Formulation}
\vspace{-2mm}
This work considers a semantic LSAT communication system consisting of one LSAT serving $K$ users (UEs) within a time window of $T$ time-slots (TSs) with TS duration of $\tau$ as depicted in Fig.~\ref{fig:system_model}. 
In this SemCom system, we focus on the image reconstruction task where the deep joint source channel encoder/decoder \cite{Gunduz_TCCN19_DJSCC} is employed.
Particularly, the LSAT employ the encoder to extract and encode semantic feature from captured images before transmitting them to the corresponding UEs. 
At the UE side, UE employed decoder to decode received semantic features and reconstruct the corresponding images.

Denote $\mathcal{L}_{k}=\{1, \dots, L_{k}\}$ as the set of $L_{k}$ images captured by the LSAT, which will be encoded and transmitted to user $k$.
For the transmission, the total system bandwidth (BW) $B_{\sf{total}}$ is divided into $S$ sub-channels (SCs), the BW of each is $B = B_{\sf{total}} / S$. Let $\mathcal{K} = \{1,\dots,K\}$, $\mathcal{S} = \{1,\dots,S\}$, and $\mathcal{T} = \{1,\dots,T\}$ be the UE, sub-channel, and TS sets; and let ${\sf{UE}}_{k}$ and ${\sf{SC}}_{s}$ denote UE $k$ and SC $s$, respectively.

\subsection{LEO-UE Association Model}
To describe the UE association, a binary variable $\boldsymbol{\alpha} = [\alpha_{k,s,t}]_{\forall (k,s,t)}$ is introducted and defined as
\begin{equation}
    \alpha_{k,s,t} = \begin{cases}
        1 & \text{ if LSAT serves } {\sf{UE}}_{k} \text{ over SC } s \text{ at TS } t, \\
        0 & \text{ otherwise.}
    \end{cases}
\end{equation}
To ensure orthogonality, each SC should be allocated to at most one UE each time, yieding 
\begin{equation}
    (C1): \quad \scaleobj{.9}{\sum_{\forall k \in \mathcal{K}} } \alpha_{k,s,t} \leq 1, \quad \forall (s,t) \in \mathcal{S} \times \mathcal{T}. \nonumber
\end{equation}
Additionally, one assumes that each UE can be allocated at most one SC, which is ensured by constraint
\begin{equation}
    (C2): \quad \scaleobj{.9}{\sum_{\forall s \in \mathcal{S}} }\alpha_{k,s,t} \leq 1, \quad \forall (k,t) \in \mathcal{K} \times \mathcal{T}. \nonumber
\end{equation}
To ensure that all UEs are served, the following constraint must be guaranteed
\begin{equation}
    (C3): \quad \scaleobj{.9}{\sum_{\forall (s,t) \in \mathcal{S} \times \mathcal{T}}} \alpha_{k,s,t} \geq 1, \quad \forall k \in \mathcal{K}. \nonumber
\end{equation}

\subsection{SemCom Model}
Let $\mathbf{x}_{k,\ell} \in \mathbb{R}^{N}$ be the $\ell-$th image corresponding to ${\sf{UE}}_{k}$. At the satellite, with input $\mathbf{x}_{k,\ell}$, the encoder extracts and encodes the semantic feature into the symbol $\mathbf{z}_{k,\ell} \in \mathbb{C}^{M}$ with $M < N$. Particularly, let $E(\cdot \vert \rho)$ be the encoder function, where $\rho = N / M$ is the compression ratio parameter. Additionally, we assume that the encoder/decoder has $V$ compression options. The compression option set with the ascending order and option $v$ are denoted by $\Pi = \{\pi_{1},\dots,\pi_{V}\}$ and $\pi_{v}$, respectively.
The output symbol for ${\sf{UE}}_{k}$ is given by 
\begin{equation}
    \mathbf{z}_{k,\ell} = E(\mathbf{x}_{k,\ell} \vert \rho_{k}), \quad \forall (k,\ell) \in \mathcal{K} \times \mathcal{L}_{k},    
\end{equation}
where $\rho_{k} \in \Pi$ is the encoder/decoder parameter decision for ${\sf{UE}}_{k}$.
For convenience, let us denote $\boldsymbol{\rho} = [\rho_{k}]_{\forall k \in \mathcal{K}}$.

Subsequently, $\mathbf{z}_{k,\ell}$ is prepared to be transmitted over wireless environment to ${\sf{UE}}_{k}$. 
Let $\bar{\mathbf{z}}_{k,t}$ be the transmitted symbol for ${\sf{UE}}_{k}$ at TS $t$ which is constructed from $\{\mathbf{z}_{k,\ell}\}_{\forall \ell}$ and normalized as $\mathbb{E}[\vert \bar{\mathbf{z}}_{k,t} \vert^2] = 1$.
The received symbol at ${\sf{UE}}_{k}$ at TS $t$ over SC $s$ is expressed as
\begin{equation}
    \bar{\mathbf{y}}_{k,s,t} =  \sqrt{p_{k,s,t}} \bar{h}_{k,s,t} \bar{\mathbf{z}}_{k,t} + \mathbf{n}_{k,s,t},
    % \bar{\mathbf{y}}_{k,s,t} =  \sqrt{\alpha_{k,s,t} p_{k,s,t}} \bar{h}_{k,s,t} \bar{\mathbf{z}}_{k,t} + \mathbf{n}_{k,s,t},
\end{equation}
where $\sqrt{p_{k,s,t}}$ and $\bar{h}_{k,s,t}$ are the transmit power and channel coefficient from LSAT to ${\sf{UE}}_{k}$ at TS $t$ over SC $s$, and $n_{k,s,t} \sim \mathcal{CN}(0, \sigma_{k}^2 \mathbf{I})$ is the additive white Gaussian noise (AWGN). The corresponding SNR is expressed as
\begin{equation}
    \gamma_{k,s,t}(\boldsymbol{p}) = p_{k,s,t} h_{k,s,t} / \sigma_{k}^2,
    % \gamma_{k,s,t}(\boldsymbol{p}) = \frac{\alpha_{k,s,t} p_{k,s,t} h_{k,s,t} }{\sigma_{k}^2},
\end{equation}
where $\boldsymbol{p} = [p_{k,s,t}]_{\forall (k,s,t)}$ and $h_{k,s,t} = \vert \bar{h}_{k,s,t} \vert^2$. 
Additionally, regarding limited on-board power budget, denoted by $p_{\sf{max}}$, and user association, the transmit power must satisfy
\begin{IEEEeqnarray}{ll}
    (C4): \quad \scaleobj{.9}{\sum_{\forall k \in \mathcal{K} } \sum_{\forall s \in \mathcal{S}}} p_{k,s,t} \leq p_{\sf{max}}, \quad \forall t, \nonumber \\
    % (C4): \quad \sum_{\forall k \in \mathcal{K} } \sum_{\forall s \in \mathcal{S}} \alpha_{k,s,t} p_{k,s,t} \leq p_{\sf{max}}, \quad \forall t. \nonumber \\
    (C5): \quad p_{k,s,t} \leq \alpha_{k,s,t} p_{\sf{max}}, \quad \forall (k,s,t). \nonumber
 \end{IEEEeqnarray}

Once $\mathbf{z}_{k,\ell}$ is transmitted entirely, the received symbols $\bar{\mathbf{y}}_{k,s,t}$ corresponding to $\mathbf{z}_{k,\ell}$ are aggregated as symbol $\mathbf{y}_{k,\ell} \in \mathbb{C}^{M}$. Subsequently, ${\sf{UE}}_{k}$ employs decoder to decode and reconstruct the corresponding image. Let $E^{-1}(\cdot \vert \rho)$ be the decoder function with parameter $\rho$, the reconstructed image is given by
\begin{equation}
    \hat{\mathbf{x}}_{k,\ell} = E^{-1}(\mathbf{y}_{k,\ell} \vert \rho_{k}).
\end{equation}
Hereafter, the reconstructed image quality is quantified by the structural similarity index metric (SSIM). Let ${\sf{SSIM}}(\mathbf{x}, \hat{\mathbf{x}})$ be the SSIM between reference image $\mathbf{x}$ and reconstructed image $\hat{\mathbf{x}}$.
% Let $f_{\sf{Q}}(\mathbf{x}, \hat{\mathbf{x}} \vert \rho, \eta)$ be the SSIM metric function with reference image $\mathbf{x}$ and reconstructed image $\hat{\mathbf{x}}$ inputs with respect to the given encoder/decoder and parameter $(\rho, \eta)$. 
Regarding QoS requirement, one assumes that ${\sf{UE}}_{k}$ requires an expected SSIM of $Q_{k}$, which is ensured as
\begin{equation}
    (C6): \quad \mathbb{E}_{\ell}[ {\sf{SSIM}}(\mathbf{x}_{k,\ell}, \hat{\mathbf{x}}_{k,\ell}) ] \geq Q_{k}, \quad \forall k \in \mathcal{K}. \nonumber
\end{equation}

\subsection{Transmisison Model}
\vspace{-2mm}
At the LSAT, $L_{k}$ images $\{\mathbf{x}_{k,\ell}\}_{\forall \ell}$ are encoded as $L_{k}$ encoded symbol sequence $\{\mathbf{z}_{k,\ell}\}_{\forall \ell}$. Hence, with a compression ratio $\rho_{k}$, there are $L_{k} N / \rho_{k} $ symbols must be transmitted to ${\sf{UE}}_{k}$. Assuming that the trasmitter and receiver hardwares are ideal with zero roll-off factor, the achievable symbol rate is $1$~symbol/s/Hz. Hence, the required number of TSs to transmit all encoded symbols to ${\sf{UE}}_{k}$ can be given by
\vspace{-2mm}
\begin{equation}
    T_{k}(\boldsymbol{\rho}) = {L_{k} N }/{(\rho_{k} B \tau)}.
\end{equation}

Regarding delay requirement, one assumes that encoded symbols of ${\sf{UE}}_{k}$ must be transmitted within $D_{k}$ TSs, which leads to convex constraint
\vspace{-2mm}
\begin{equation}
    (C7): \quad \scaleobj{.9}{\sum_{t=1}^{D_{k}} \sum_{\forall s \in \mathcal{S}}} \alpha_{k,s,t} \geq T_{k}(\boldsymbol{\rho}), \quad \forall k \in \mathcal{K}. \nonumber
\end{equation}

\subsection{Problem Formulation}
Regarding the limited LSAT power, this work aims to jointly optimize UE association, power control, and SemCom JSCC encoder/decoder parameter decisions for minimizing the transmit power at LSAT while meeting QoS of SSIM requirement at the users. Particularly, given the inherent LSAT movement and the time-varying channel, the images should be encoded and transmitted appropriately and the encoded symbols should be transmitted at appropriate TSs instead of immediate transmission. The optimization problem can be mathematically formulated as
\vspace{-2mm}
\begin{IEEEeqnarray}{ll}
    (\mathcal{P}_{0}): \quad \min_{\boldsymbol{p}, \boldsymbol{\alpha}, \boldsymbol{\rho}} \quad & F_{\sf{obj}}(\boldsymbol{p}) \triangleq \scaleobj{.9}{\sum\nolimits_{\forall (k,s,t)}} p_{k,s,t}  \quad \st \; (C1)-(C7), \nonumber \\
    & (C0): \rho_{k} \in \Pi, \alpha_{k,s,t} \in \{0,1\}, \quad \forall (k,s,t),\nonumber
\end{IEEEeqnarray}
with $\boldsymbol{\alpha} = [\alpha_{k,s,t}]_{\forall (k,s,t)}$. Problem $(\mathcal{P}_{0})$ is challenging to solve due to the coupling between discrete and continuous variables, non-convex functions, expected SSIM requirement in $(C4)$, and especially the lack of analytical relationship among SSIM metric, compresion ratio, and SNR.

% ===========================================
\section{Proposed Solution}
\begin{figure}
    \centering
    \captionsetup{font=small}
    \includegraphics[width=1\linewidth]{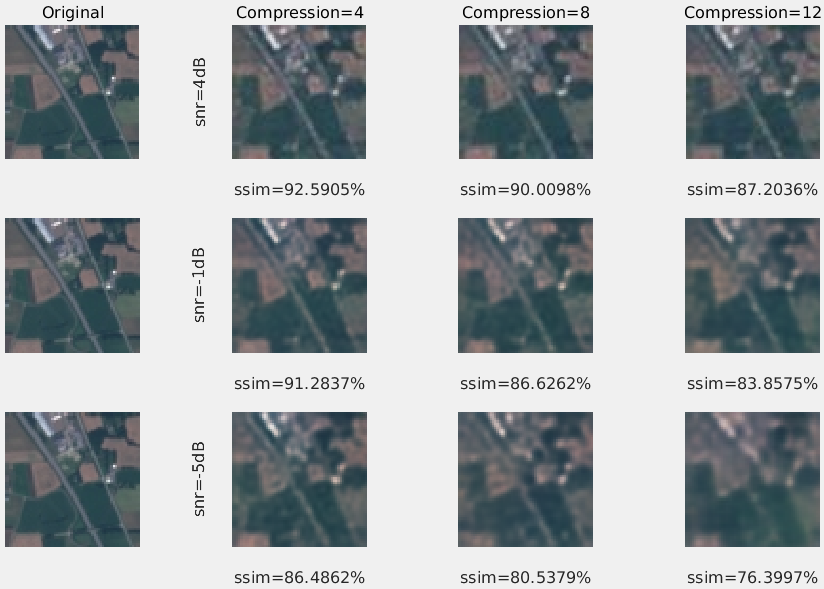}
    \caption{Reconstruction quality versus compression ratio and SNR.}
    \label{fig:reconstruction_example}
    \vspace{-0.2cm}
\end{figure}
\begin{figure}
    \centering
    \captionsetup{font=small}
    \includegraphics[width=1\columnwidth]{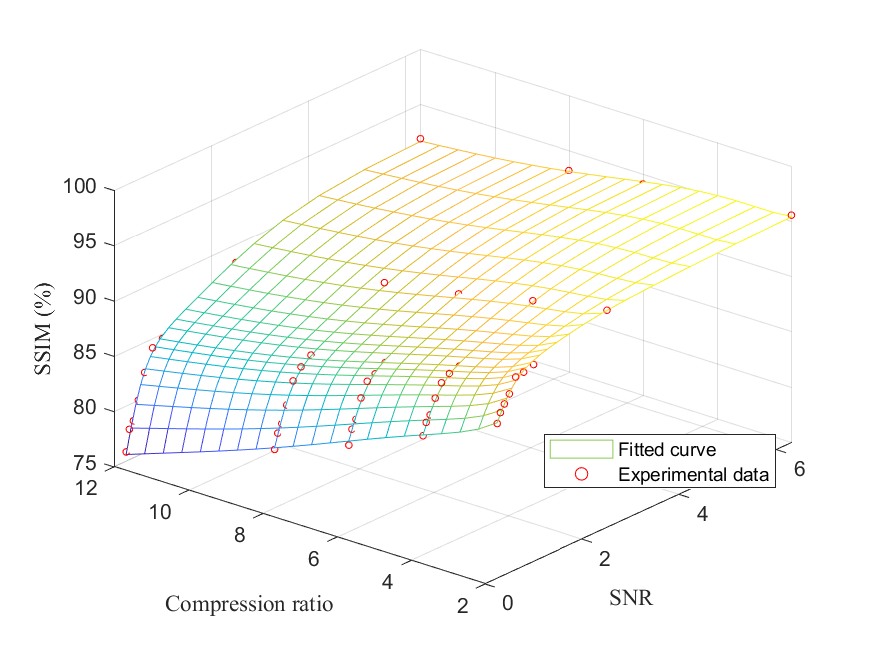}
    \vspace{-0.5cm}
    \caption{Curve-Fitting-based Approximation.}
    \label{fig:curve_fitting}
    \vspace{-0.2cm}
\end{figure}

\begin{table}
    \captionsetup{font=small}
    \centering
    \caption{Curve fitting coefficients.}
    \label{tab:curvefit}
    \scalebox{0.8}{
    \begin{tabular}{|l|c|c|c|c|c|c|}
         \hline
        $j$ & $a_{0,j}$ & $a_{1,j}$ & $a_{2,j}$ & $a_{3,j}$ & $a_{4,j}$ & $\bar{a}$\\ 
        \cline{1-6}
        \hline
        1 & 126.216875 & -15.5512520 & -1.740186 & 19.306749 & -0.308751 & \multirow{6}{*}{30.429844}\\
        2 & 59.985122 & -6.605492 & -0.996148 & 1.831855 & 0.028218 &\\
        3 & 153.050377 & 4.196036 & -0.213359 & -7.217142 & -0.701829 & \\
        4 & -206.696783 & 9.487473 & -1.375091 & 3.405617 & -0.549287 & \\
        5 & -38.0243540 & 6.00383 & 0.663762 & 5.409884 & -0.339842 & \\
        6 & -15.4042140 & 6.49103 & 4.730506 & -1.064734 & 0.0851 & \\
         \hline
    \end{tabular}}
\end{table}

\subsection{Curve-Fitting-based Approximation}
In this work, the JSCC encoder and decoder proposed in \cite{Gunduz_TCCN19_DJSCC} are utilized where the EUROSAT dataset is employed for traning and validation \cite{EUROSAT_dataset}. Fig~.\ref{fig:reconstruction_example} shows an example of image reconstruction quality versus compression ratio and SNR levels.
According to trained results, expected SSIM is affected by both compression ratio $\rho$ and SNR $\gamma$ as depicted in Fig.~\ref{fig:curve_fitting}. 
To solve problem $\mathcal{P}_{0}$, the relationship among these metrics should be modeled mathematically.
By utilizing the curve-fitting method \cite{TiNguyen_ICMLCN2025, nguyen2026unifiedsemanticlossmodel}, the reconstruction quality in terms of expected SSIM can be approximated as a function of the compression ratio and transmission SNR as
\vspace{-2mm}
% \begin{equation}\label{eq: SSIM fitting}
%     f_{\sf{Q}}(\rho, \gamma) = a_{0} + a_{1} \log_{10}(a_{2} + a_{3} \gamma) + a_{4} \log_{10}(a_{5} + a_{6} \rho),
% \end{equation}
\begin{equation}\label{eq: SSIM fitting}
    f_{\sf{Q}}(\rho, \gamma) = \bar{a} + \scaleobj{.9}{\sum_{j=1}^{J}}(a_{0,j} \exp(a_{4,j} \rho) + a_{1,j} f_{{\sf{es}},j}(\rho, \gamma) ),
\end{equation}
where $f_{{\sf{es}},j}(\rho, \gamma) = \frac{\exp(a_{4,j}\rho)}{1 + \exp(-a_{2,j} \gamma - a_{3,j})}$ while $\bar{a}, a_{x,j}$, are the fitting coefficients expressed in Table~\ref{tab:curvefit}. The fitting model is shown in Fig.~\ref{fig:curve_fitting}.
Employing \eqref{eq: SSIM fitting}, problem $(\mathcal{P}_{0})$ can be approximated as
\begin{IEEEeqnarray}{lcl}
    (\mathcal{P}_{1}): &    \min_{\boldsymbol{p}, \boldsymbol{\alpha}, \boldsymbol{\rho}, \boldsymbol{\eta}} \quad & F_{\sf{obj}}(\boldsymbol{p}) \nonumber \\
    & \st \; & (C0)-(C5), (C7), \nonumber \\
    &&(\hat{C}6a): \quad  f_{\sf{Q}}(\rho_{k}, \eta_{k}) \geq Q_{k}, \quad \forall k \in \mathcal{K}, \nonumber \\
    &&(\hat{C}6b): \quad  \gamma_{k,s,t}(\boldsymbol{p}) \geq \alpha_{k,s,t} \eta_{k}, \quad \forall (k,s,t), \nonumber 
\end{IEEEeqnarray}
where $\boldsymbol{\eta} = [\eta_{k}]_{\forall k}$, $\eta_{k}$ is introduced as target SNR of ${\sf{UE}}_{k}$ which ensure quality constraint $(C6)$. Constraints $(\hat{C}6a)$ and $(\hat{C}6b)$ are the approximated form of $(C6)$. Particularly, $(\hat{C}6b)$ ensures that once ${\sf{UE}}_{k}$ is served at TS $t$ over SC $s$, the compression ratio $\rho_{k}$ and corresponding SNR $\gamma_{k,s,t}(\boldsymbol{p})$ guarantee the SSIM requirement. Problem $(\mathcal{P}1)$, however, is still challenging to solve due to non-convex constraint $(\hat{C}6)$ and coupling of discrete-continuous variables. Problem $(\mathcal{P}1)$ is simplified by the following theorem.

\begin{theorem} \label{theorem: transform prob}
    Problem $(\mathcal{P}_{1})$ can be transformed into an equivalent problem $(\mathcal{P}_{2})$ (has the same optimal solution) as
    \begin{IEEEeqnarray}{lcl}
        (\mathcal{P}_{2}): &    \min_{\boldsymbol{p}, \boldsymbol{\alpha}, \boldsymbol{\rho}} \quad & F_{\sf{obj}}(\boldsymbol{p})\nonumber \\
        & \st \; & (C0)-(C5), (C7), \nonumber \\
        && (\breve{C}6): \; f_{\sf{Q}}(\rho_{k}, \gamma_{k,s,t}(\boldsymbol{p})) \geq \alpha_{k,s,t} Q_{k}, \quad \forall (k,s,t), \nonumber
    \end{IEEEeqnarray}
\end{theorem}
\begin{IEEEproof}
    Please see Appendix~\ref{appendix: transform prob}
\end{IEEEproof}

\subsection{Discrete Relaxation and Problem Approximation}
In problem $(\mathcal{P}_{2})$, constraint $\breve{C}6$ is non-convex due to nonconvexity of $f_{\sf{Q}}(\rho,\gamma)$ and discrete variables $(\boldsymbol{\alpha}, \boldsymbol{\rho})$.  
To transform $(\mathcal{P}_{2})$ into a more traceable form, discrete variables $\boldsymbol{\alpha}$ and $\boldsymbol{\rho}$ are relaxed into a continuous ones as
\begin{IEEEeqnarray}{ll}
    (\tilde{C}0a): \quad 0 \leq \alpha_{k,s,t} \leq 1, \quad \forall (k,s,t), \nonumber \\
    (\tilde{C}0b): \quad \pi_{1} \leq \rho_{k} \leq \pi_{V}, \quad \forall k. \nonumber
\end{IEEEeqnarray}

%Subsequently, constraint $(\breve{C}6)$ is approximated by the following proposition.
\begin{proposition} \label{pro: convex C6}
    Constraint $(\breve{C}6)$ is convexified as
    \begin{IEEEeqnarray}{ll}
        (\tilde{C}6a): \; 1 + e^{-a_{2,j} \gamma_{k,s,t} - a_{3,j}} \leq u_{k,s,t,j},  \forall (k,s,t),\forall j \vert a_{1,j} \geq 0, \nonumber \\
        (\tilde{C}6b): \; \bar{a} + \scaleobj{.9}{\sum_{j=1}^{J}} \Big[ a_{0,j} \Big( \scaleobj{.9}{1_{\{a_{0,j} \geq 0\}} f_{\sf{exp}}^{(i)}(\rho_{k};a_{4,j},0)  + 1_{\{a_{0,j} < 0\}} e^{a_{4,j} \rho_{k}}} \Big) \nonumber \\
         \hspace{5mm}
        + a_{1,j} \Big( \scaleobj{.9}{1_{\{ a_{1,j} \geq 0 \}} f_{{\sf{es}},j}^{{\sf{low}}, (i)}\!\rho_{k}, u_{k,s,t,j})\! +\! 1_{\{ a_{1,j} < 0 \}} f_{{\sf{es}},j}^{{\sf{up}}, (i)}\!(\rho_{k}, u_{k,s,t,j})} \Big) \Big]  \nonumber \\
        \hspace{55mm} 
         \geq \alpha_{k,s,t} Q_{k}, \; \forall (k,s,t), \nonumber
    \end{IEEEeqnarray}
    where $\boldsymbol{u}=[u_{k,s,t,j}]$ is a new slack variable, $f_{{\sf{exp}}}^{(i)}(x;a,b) = a e^{a x^{(i)}+b}(x - x^{(i)} + 1/a)$, $f_{{\sf{es}},j}^{{\sf{low}},(i)}(\rho, u) =  e^{a_{4,j}\rho^{(i)}}/{u^{(i)}}(a_{4,j}\rho - u/u^{(i)} - a_{4,j}\rho^{(i)} + 2)$, and $f_{{\sf{es}},j}^{{\sf{up}},(i)}(\rho, \gamma) =  e^{a_{4,j}\rho} / (1 + f_{\sf{exp}}^{(i)}(\gamma; -a_{2,j}, -a_{3,j}))$.
\end{proposition}
\begin{IEEEproof}
    Please see appendix~\ref{appendix: convex C6}
\end{IEEEproof}

Thanks to Proposition~\ref{pro: convex C6}, $(\mathcal{P}_{2})$ is transformed into 
\begin{IEEEeqnarray}{lcl}
    (\mathcal{P}_{3}): &    \quad \min_{\boldsymbol{p}, \boldsymbol{\alpha}, \boldsymbol{\rho}, \boldsymbol{u}} \quad & F_{\sf{obj}}(\boldsymbol{p}) \; \st \; (\tilde{C}0),(C1)-(C5), (\tilde{C}6), (C7). \nonumber
\end{IEEEeqnarray}
It is worth noting that solving $(\mathcal{P}_{3})$ returns continuous values of $\boldsymbol{\rho}$ and $\boldsymbol{\alpha}$. Then, their discrete values can be recovered as
\begin{IEEEeqnarray}{ll}
    \rho_{k} = \mathtt{argmin}_{\bar{\pi} \in \Pi} \vert \bar{\pi} - \rho_{k} \vert, \quad \forall k, \subnum \label{eq: recover compression} \\
    \alpha_{k,s,t} = 1 \text{ if } \alpha_{k,s,t} \geq 0.5 \text{ and } \alpha_{k,s,t} = 0 \text{ if } \alpha_{k,s,t} < 0.5. \quad\quad \subnum \label{eq: recover binary}
\end{IEEEeqnarray}
The proposed algorithm is summarized in Algorithm~\ref{alg: JCRRA}, namely joint compression ratio-RA \textbf{(JCRRA)}.

\begin{algorithm}[!t]
	\footnotesize
	\begin{algorithmic}[1]
		\captionsetup{font=small}
		\protect\caption{\textsc{JCRRA Algorithm for Solving Problem $(\mathcal{P}_{0})$}}
		\label{alg: JCRRA}
		% \long\def\algorithmicrequire{\textbf{Phase 1:}}
		% \REQUIRE
		\STATE Generate feasible point $(\boldsymbol{p}^{(i)}, \boldsymbol{\rho}^{(i)}, \boldsymbol{u}^{(i)})$ and set $i=1$
		\REPEAT
		\STATE Solve problem $(\mathcal{P}_{3})_{\kappa}$ to obtain $(\boldsymbol{p}^\star, \boldsymbol{\alpha}^\star, \boldsymbol{\rho}^\star, \boldsymbol{u}^\star)$.
        \STATE Update $(\boldsymbol{p}^{(i)}, \boldsymbol{\rho}^{(i)}, \boldsymbol{u}^{(i)}) = (\boldsymbol{p}^\star, \boldsymbol{\rho}^\star, \boldsymbol{u}^\star)$ and $i=i+1$
		\UNTIL Convergence
		\STATE Recovery compression solution $\boldsymbol{\rho}$ by \eqref{eq: recover compression} and repeat steps 2-5
		\STATE Recovery binary solution $\boldsymbol{\alpha}$ by \eqref{eq: recover binary} and repeat steps 2-5
		\STATE \textbf{Output:} The RA and compression solutions $(\boldsymbol{p}^{\star}, \boldsymbol{\alpha}^{\star}, \boldsymbol{\rho}^{\star})$
	\end{algorithmic}
	\normalsize 
\end{algorithm}

\subsection{Other Benchmarks}
For comparison purpose, two greedy algorithms for JSCC-based and JPEG-compression-based systems are proposed.
\subsubsection{JSCC-based Greedy Algorithm}
For a given compression ratio $\rho_{k} = \bar{\rho}$, user association is selected based on channel gain so that number of connections satisfies $(C7)$. Subsequently, the minimum target SNR level $\eta_{k}$ and transmit power are identified to satisfy $(\hat{C}6)$. This algorithm is described in Algorithm.~\ref{alg: Greedy JSCC}, namely \textbf{Greedy JSCC}.
\subsubsection{JPEG-Compression-based Greedy Algorithm}
Given SSIM requirement, images are compressed in JPEG format with a corresponding quality before transmisison. First, user association is selected based on channel gain in desending order. By utilizing the Shanon capacity fomular to compute data rate for transmisison of compressed image files and the waterfilling power allocation, the minimum transmit power is identified. This algorithm is described in Algorithm~\ref{alg: Greedy JPEG}, namely \textbf{Greedy JPEG}.

\begin{algorithm}[!t]
	\footnotesize
	\begin{algorithmic}[1]
		\captionsetup{font=small}
		\protect\caption{\textsc{Greedy JSCC Algorithm for Solving $(\mathcal{P}_{0})$}}
		\label{alg: Greedy JSCC}
		\STATE \textbf{Input:} Channel gain matrix $\mathbf{h}$, fixed compression $\bar{\rho}$, and quality $Q_{k}$
		\STATE \textbf{Initialize:} Set $\boldsymbol{\alpha} = \boldsymbol{0}$, number of required connections $\bar{L}_{k} = L_{k} N/(\bar{\rho} B \tau)$, and channel after TS $D_{k}$ be zero $\mathbf{h}(k,:,D_{k}+1) = 0$
        \REPEAT
            \STATE Find index of maximum element of $\mathbf{h}$: $(\hat{k}, \hat{s}, \hat{t})$
            \STATE Set $\alpha_{\hat{k}, \hat{s}, \hat{t}} = 1$, $\mathbf{h}(\hat{k},:,\hat{t}) = 0$, $\mathbf{h}(:,\hat{s},\hat{t}) = 0$, and $\bar{L}_{\hat{k}} = \bar{L}_{\hat{k}} - 1$
            \STATE \textbf{If} $\bar{L}_{\hat{k}} = 0$ \textbf{do} set $\mathbf{h}(\hat{k},:,:)=0$
        \UNTIL  $\mathbf{h} = \boldsymbol{0}$ or $\bar{L}_{k} = 0, \forall k$
        \STATE Identify target $\eta_{k}| f_{\sf{Q}}(\bar{\rho}, \eta_{k}) = Q_{k}$ and $p_{k,s,t} = \alpha_{k,s,t} \eta_{k}$
        \STATE \textbf{Output:} User association and transmit power solutions $\boldsymbol{\alpha}$ and $\boldsymbol{p}$
	\end{algorithmic}
	\normalsize 
\end{algorithm}

\begin{algorithm}[!t]
	\footnotesize
	\begin{algorithmic}[1]
		\captionsetup{font=small}
		\protect\caption{\textsc{Greedy JPEG Algorithm for Solving $(\mathcal{P}_{0})$}}
		\label{alg: Greedy JPEG}
		\STATE \textbf{Input:} Channel gain matrix $\mathbf{h}$ and quality $Q_{k}$
        \STATE \textbf{For each} $k$: Compress images in JPEG format satisfying SSIM $Q_{k}$
		\STATE Set $\boldsymbol{\alpha} = \boldsymbol{0}$ and channel after TS $D_{k}$ be zero $\mathbf{h}(k,:,D_{k}+1) = 0$
        \REPEAT
            \STATE Find index of maximum element of $\mathbf{h}$: $(\hat{k}, \hat{s}, \hat{t})$
            \STATE Set $\alpha_{\hat{k}, \hat{s}, \hat{t}} = 1$, $\mathbf{h}(\hat{k},:,\hat{t}) = 0$, $\mathbf{h}(:,\hat{s},\hat{t}) = 0$
        \UNTIL Channel gain matrix is zero $\mathbf{h} = \boldsymbol{0}$
        \FOR{Each $k$, initialize water level}
            \REPEAT
            \STATE Allocate power for selected links using waterfilling algorithm 
            \STATE Compare sum rate with total image size. Update water level
            \UNTIL $B \tau \sum_{\forall (i,j)} \log_{2}(1 + h_{k,i,j}p_{k,i,j} / \sigma_{k}) \!\! \approx \!\! \sum [\mathsf{image \; size \; of \; user \; k}]$
        \ENDFOR
        \STATE \textbf{Output:} User association and transmit power solutions $\boldsymbol{\alpha}$ and $\boldsymbol{p}$
	\end{algorithmic}
	\normalsize 
\end{algorithm}

% ===========================================
\vspace{-2mm}
\section{Numerical Results}
\vspace{-2mm}
\begin{table}[t]
	\caption{Simulation parameters.}
	\label{tab:parameter}
	\centering
    \scalebox{0.8}{
	\begin{tabular}{l|l}
		\hline
		Parameter & Value \\
		\hline\hline
        Downlink Frequency & $20$ GHz \\
        Bandwidth, $B_{\sf{total}}$   &  40 MHz \\
        Time-slot duration $\tau$ & $1$~second \\
        JSCC compression ratio options, $\Pi$ & $\{48/4, 48/5, \dots, 48/24 \}$ \\
        LSAT orbit                  & Starlink orbit at $550$ km \\
        LSAT antenna gain,    & 24 dBi \cite{3gpp.38.821} \\
		Noise power power density,	& -174 dBm/Hz \\
		LSAT power budget, $ p_{\sf{max}} $	& 17 dBW \\
        Dataset for experiment results & EUROSAT \cite{EUROSAT_dataset} \\
		Number of usesrs, sub-channels, time-slots, $(K,S,T) $ & $(5, 4, 100)$ \\
        Number of images for users, $[L_{k}]$  & $[60, 65, 70, 75, 80] \times 10^3$ \\
        SSIM requirement of users, $[Q_{k}]$  & $[80, 83, 85, 87, 90] $ \% \\
        Delay requirement of users, $[D_{k}]$  & $[65, 60, 55, 55, 45]$ time slots \\
		\hline		   				
	\end{tabular}
    }
	%	\vspace{-0.5pt}
\end{table}

In this section, the numerical results are studied to evaluate the proposed algorithms in various scenarios. The EUROSAT dataset \cite{EUROSAT_dataset} is used for LSAT EO simulation. The simulated system is examined in a time-window of $100$~seconds. The channel coefficients are generated by incorporating the satellite movement, channel fading, and terminal antenna gain. The key simulation parameters are described in Table~\ref{tab:parameter}.

\begin{figure}[!t]
    \centering
    \captionsetup{font=small}
    \includegraphics[width=0.9\columnwidth]{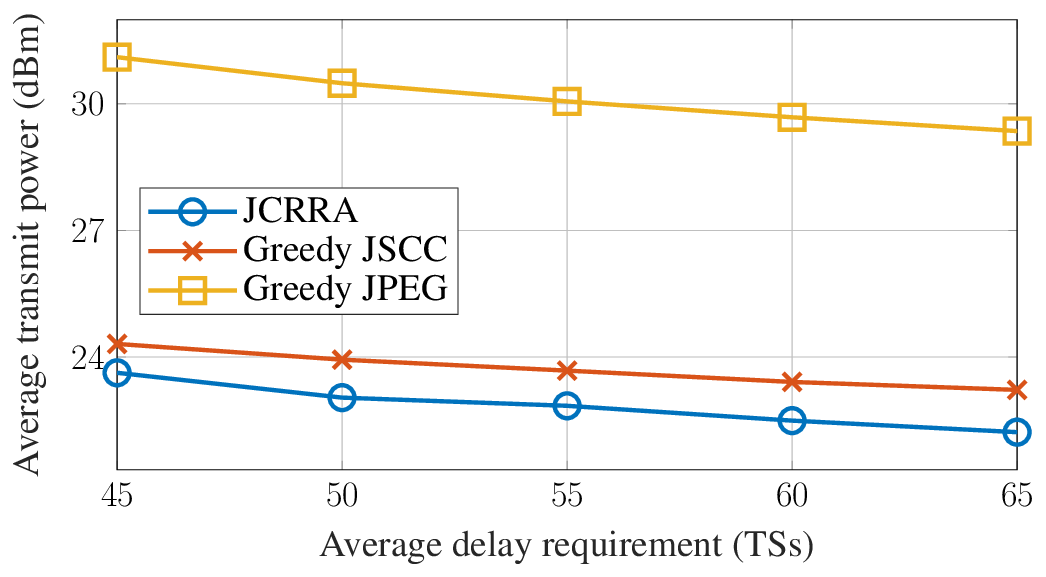}
    \vspace{-3mm}
    \caption{Average transmit power versus delay requirement.}
    \label{fig:P_delay}
    \vspace{-3mm}
\end{figure}
Fig.~\ref{fig:P_delay} shows the average transmit power versus the delay requirement. In general, to meet a stricter delay requirement, LSAT must transmit with higher power. In particular, when the average delay requirement increases $\bar{D} = 45 \rightarrow 65$ TSs, the transmit power of all three schemes decreases by approximately $1.2-1.8$ dBm. This phenomenon arises from the time-varying channel gain--affected by factors such as slant range and antenna beam gain--caused by the LSAT's orbital movement. Consequently, with a more relaxed delay requirement, the LSAT has greater flexibility to transmit symbols during TSs with higher channel gains, thereby reducing the required transmit power. Furthermore, compared with the conventional transmission method employing JPEG compression, the JSCC-based transmission demonstrates a significant power-saving advantage in EO applications. In particular, under similar greedy strategies, the performance gap between the Greedy JSCC and JPEG schemes is about $6$~dBm. Moreover, by jointly optimizing compression and RA, the proposed JCRRA algorithm achieves an additional power reduction of approximately $1.1$~dBm.

\begin{figure}[!t]
    \centering
    \captionsetup{font=small}
    \includegraphics[width=0.9\columnwidth]{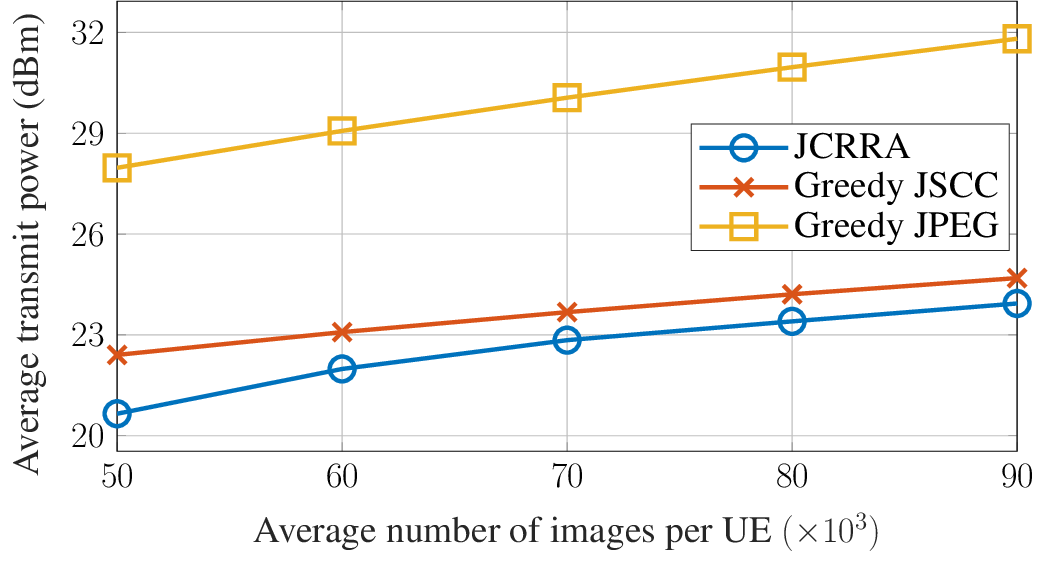}
    \vspace{-3mm}
    \caption{Average transmit power versus number of images per user.}
    \label{fig:P_length}
    \vspace{-3mm}
\end{figure}
Fig.~\ref{fig:P_length} depicts the average transmit power as a function of the mean number of demanded images per user. As expected, the LSAT requires higher total transmit power to accommodate a larger image demand. Particularly, In the Greedy JPEG scheme, the images are first compressed into bitstreams for transmission; thus, a larger demand results in a higher total data volume, which in turn necessitates a higher data rate and transmit power. For the JSCC-based schemes, an increase in the number of images necessitates more TSs to transmit all symbols, thereby resulting in more power consumption. By directly encoding images into transmission symbols, the JSCC method reduces significant power consumption. Specifically, comparing two greedy schemes, the JSCC approach achieves a power reduction of approximately $5-7$~dBm across the considered image demand scenarios. Furthermore, thanks to adaptive compression selection and RA, the proposed JCRRA algorithm achieves an additional power saving of about $0.7-1.8$~dBm compared with the Greedy JSCC scheme.

\begin{figure}[!t]
    \centering
    \captionsetup{font=small}
    \includegraphics[width=0.9\columnwidth]{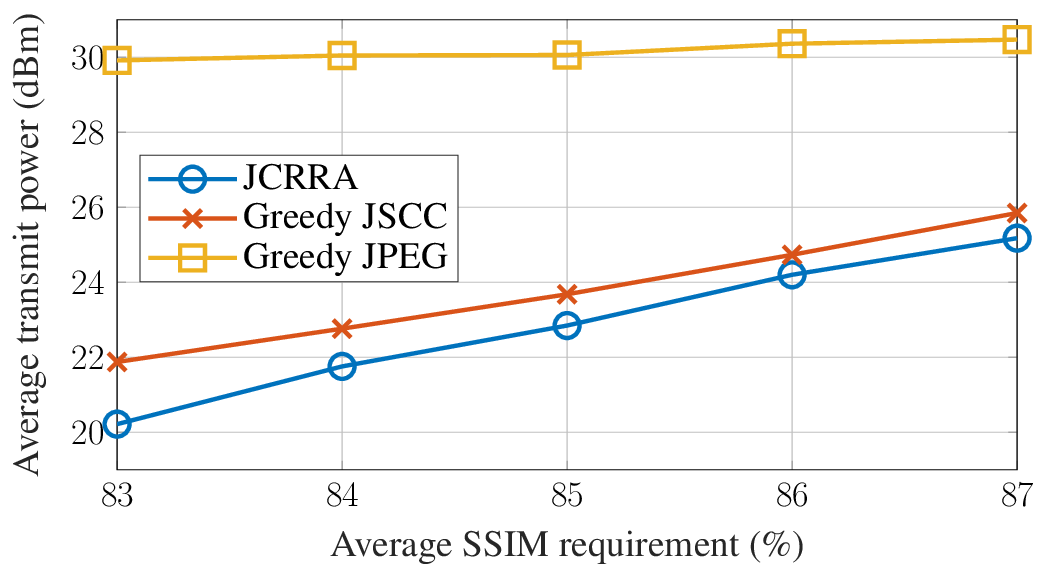}
    \vspace{-3mm}
    \caption{Average transmit power versus SSIM requirement.}
    \label{fig:P_ssim}
    \vspace{-3mm}
\end{figure}

Fig.~\ref{fig:P_ssim} illustrates the average transmit power under different mean SSIM requirements. As expected, a lower SSIM requirement corresponds to a lower required SNR level, thereby reducing transmit power consumption. However, in the Greedy JPEG scheme, the decrease in transmit power with respect to examined SSIM is relatively slow, since lowering the SSIM requirement only slightly reduces the compressed data size, i.e, decreasing SSIM requirement from $87$\% to $83$\% yields a power saving of less than $1$~dBm. In contrast, the JSCC-based schemes exhibit a more substantial reduction in power consumption, particularly at lower SSIM levels. This trend is consistent with the SSIM-compression-SNR relationship shown in Fig.~\ref{fig:curve_fitting}, where the required SNR drops sharply once SSIM falls below approximately $87$\%. Consequently, the power consumption gap between the JSCC and JPEG schemes becomes more pronounced at lower SSIM requirements.

% ===========================================
\vspace{-2mm}
\section{Conclusion}
\vspace{-2mm}
In conclusion, this work investigated a SemCom framework incorporating JSCC for LSAT EO systems. The proposed optimization framework was efficiently addressed through an iterative algorithm that integrates a curve-fitting model and the SCA approach. Numerical results demonstrate the effectiveness of employing JSCC in reducing the transmit power for satellite EO systems. In particular, the proposed algorithm consistently outperforms two benchmark schemes in terms of power minimization, highlighting the potential of JSCC-based SemCom as an energy-efficient solution for future satellite communication systems.

% ===========================================

\appendices
\vspace{-2mm}
\section{Proof of Theorem~\ref{theorem: transform prob}} \label{appendix: transform prob}
\vspace{-2mm}
First, one notes that $\gamma_{k,s,t}(\boldsymbol{p})$ and $f_{\sf{Q}}(\rho_{k}^{\star}, \eta_{k})$ are monotonic increasing with respect to $p_{k,s,t}$, respectively. Hence, at the optimal point of both problems, if $\alpha_{k,s,t}=0$, $p_{k,s,t}=0$ to minimize power.
Additionally, it can be seen that once $(\breve{C}6)$ (with $\alpha_{k,s,t}=1$) and $(\hat{C}6a)$, $(\hat{C}6b)$ hold equality at the optimal point of $(\mathcal{P_{2}})$ and $(\mathcal{P}_{1})$, respectively, two problems are equivalent, i.e., the optimal point of $(\mathcal{P_{1}})$ is a feasible point of $(\mathcal{P_{2}})$ and vice versa. 

Let's consider $(\mathcal{P}_{1})$ with an optimal point $(\boldsymbol{p}^{\star}, \boldsymbol{\alpha}^{\star}, \boldsymbol{\rho}^{\star}, \boldsymbol{\eta}^{\star})$. Assuming that $(\boldsymbol{p}^{\prime}, \boldsymbol{\alpha}^{\star}, \boldsymbol{\rho}^{\star}, \boldsymbol{\eta}^{\prime})$ is a feasible point of $(\mathcal{P}_{1})$, 
if $\exists (k,s,t) \vert f_{\sf{Q}}(\rho_{k}^{\star}, \eta_{k}^{\prime}) > Q_{k}$ and $\gamma_{k,s,t}(\boldsymbol{p}^{\prime}) > \alpha_{k,s,t}^{\star} \eta_{k}^{\prime}$, there exists a better solution $p_{k,s,t}^{\prime \prime} < p_{k,s,t}^{\prime}$ and $\eta_{k}^{\prime \prime} < \eta_{k}^{\prime}$ satisfying $(\hat{C}6a),(\hat{C}6b)$ with $F_{\sf{obj}}(\boldsymbol{p}^{\prime \prime}) < F_{\sf{obj}}(\boldsymbol{p}^{\prime})$. Hence, at optimal point $(\boldsymbol{p}^{\star}, \boldsymbol{\alpha}^{\star}, \boldsymbol{\rho}^{\star}, \boldsymbol{\eta}^{\star})$, $(\hat{C}6a)$ and $(\hat{C}6b)$ hold equality. Additionally, $(\boldsymbol{p}^{\star}, \boldsymbol{\alpha}^{\star}, \boldsymbol{\rho}^{\star})$ is a feasible point of $(\mathcal{P}_{2})$.

Similarly, let's consider $(\mathcal{P}_{2})$ with an optimal point $(\boldsymbol{p}^{\star\star}, \boldsymbol{\alpha}^{\star\star}, \boldsymbol{\rho}^{\star\star})$. Assuming that $(\boldsymbol{p}^{\prime}, \boldsymbol{\alpha}^{\star\star}, \boldsymbol{\rho}^{\star\star})$ is a feasible point, if $\exists (k,s,t) \vert f_{\sf{Q}}(\rho_{k}^{\star\star}, \gamma_{k,s,t}(\boldsymbol{p}^{\prime})) > \alpha_{k,s,t}^{\star\star} Q_{k}$, there exists $p_{k,s,t}^{\prime\prime} < p_{k,s,t}^{\prime}$ satisfying $(\breve{C}6)$ with $F_{\sf{obj}}(\boldsymbol{p^{\prime\prime}}) < F_{\sf{obj}}(\boldsymbol{p^{\prime}})$. Hence, at optimal point $(\boldsymbol{p}^{\star\star}, \boldsymbol{\alpha}^{\star\star}, \boldsymbol{\rho}^{\star\star})$, $(\breve{C}6)$ holds equality. Besides, $(\boldsymbol{p}^{\star\star}, \boldsymbol{\alpha}^{\star\star}, \boldsymbol{\rho}^{\star\star})$ is also feasible with $(\mathcal{P}_{1})$.

\vspace{-2mm}

\section{Proof of Proposition~\ref{pro: convex C6}} \label{appendix: convex C6}
To convexify $(\breve{C}6)$, one needs to find a lower bound of $f_{\sf{Q}}(\rho_{k},\gamma_{k,s,t})$. Hence, either the lower or upper bounds of components must be identified for positive or negative coefficients ($a_{0,j}$ and $a_{1,j}$), respectively. 
Applying the first-order approximation \cite{boyd2004convex} 
the lower bounds of $f_{\sf{exp}}(x;a,b) = e^{ax + b}$ and $\hat{f}_{{\sf{es}},j}(\rho, u) = e^{a_{4,j} \rho} / u$ can be expressed as \vspace{-2mm}
\begin{IEEEeqnarray}{ll}
    \hspace{-3mm} f_{{\sf{exp}}}(x;a,b) \!\! \geq \! \! \nabla \! f_{{\sf{exp}}}\vert_{x^{(i)}} \! 
    (\! x \!-\! x^{(i)} \!) \!+\! f_{\sf{exp}}(x^{(i)} \! ; a, \! b) \!\! = \!\!  f_{\sf{exp}}^{(i)}(x ; \! a, \!b), \label{eq: fexp0} \\
    \hspace{-3mm} \hat{f}_{{\sf{es}},j}(\rho,u) \geq f_{{\sf{es}},j}(\rho^{(i)},u^{(i)}) + [\nabla_{\rho} f_{{\sf{es}},j}, \nabla_{u} f_{{\sf{es}},j}] \vert_{(\rho^{(i)},u^{(i)})} \nonumber \\
    \hspace{26mm} \times [\rho - \rho^{(i)}, u - u^{(i)}]^{T} = f_{{\sf{es}},j}^{(i)}(\rho,u). \label{eq: fes0}
\end{IEEEeqnarray}
Regarding $e^{a_{4,j} \rho_{k}}$ term, when $a_{0,j} \geq 0$, the following is obtained by applying \eqref{eq: fexp0} as 
\begin{equation}
    e^{a_{4,j} \rho_{k}} \geq f_{\sf{exp}}^{(i)}(\rho_{k};a_{4,j},0).   
\end{equation}
For component $f_{{\sf{es}},j}(\rho_{k}, \gamma_{k,s,t})$ in case $a_{1,j} \geq 0$, with $\boldsymbol{u}$ satisfying $(\tilde{C}6a)$, a lower bound of $f_{\sf{Q}}(\rho_{k}, \gamma_{k,s,t})$ is obtained by applying \eqref{eq: fes0} as  \vspace{-2mm}
\begin{equation} \label{eq: fes low}
    \hspace{-2mm} f_{{\sf{es}},j}(\rho_{k}, \gamma_{k,s,t}) \! \geq \!\hat{f}_{{\sf{es}},j}(\rho_{k}, u_{k,s,t,j}) \! \geq \! f_{{\sf{es}},j}^{{\sf{low}},(i)}(\rho_{k},u_{k,s,t,j}). 
\end{equation}
% One can see that $f_{{\sf{apx}},j}(\rho,u)$ is a convex function over $\rho > 0$ and $u > 0$ since the corresponding Hessian matrix is positive definite. Apply first-order approximation \cite{boyd2004convex}, its lower bound is expressed as
% \begin{IEEEeqnarray}{ll} \label{eq: fapx}
%     f_{{\sf{apx}},j}(\rho,u) & \geq f_{{\sf{apx}},j}(\rho^{(i)},u^{(i)}) + [\nabla_{\rho} f_{{\sf{apx}},j}, \nabla_{u} f_{{\sf{apx}},j}] \vert_{(\rho^{(i)},u^{(i)})} \nonumber \\
%     & \quad \quad \times [\rho - \rho^{(i)}, u - u^{(i)}]^{T} = f_{{\sf{apx}},j}^{(i)}(\rho,u).
% \end{IEEEeqnarray}
When $a_{1,j} < 0$, employing \eqref{eq: fexp0} again for $e^{-a_{2,j} \gamma_{k,s,t} - a_{3,j}}$ terms in $f_{{\sf{es}},j}(\rho_{k}, \gamma_{k,s,t})$ yields \vspace{-2mm}
\begin{equation}
    e^{-a_{2,j} \gamma_{k,s,t} - a_{3,j}}
    \geq f_{\sf{exp}}^{(i)}(\gamma_{k,s,t}; -a_{2,j}, -a_{3,j}).
\end{equation}
Subsequently, upper bound $f_{{\sf{es}},j}^{{\sf{up}},(i)}(\rho_{k}, \gamma_{k,s,t})$ of $f_{{\sf{es}},j}(\rho_{k}, \gamma_{k,s,t})$ can be easily obtained. Finally,
% Similarly, a lower bound of $f_{{\sf{exp}},j}(\rho)= a_{0,j} \exp(a_{4,j}\rho)$ is 
% \begin{equation} \label{eq: fexp}
%     f_{{\sf{exp}},j}(\rho) \geq \nabla f_{{\sf{exp}},j}\vert_{\rho^{(i)}} (\rho - \rho^{(i)}) + f_{{\sf{exp}},j}(\rho^{(i)}) = f_{{\sf{exp}},j}^{(i)}(\rho).
% \end{equation}
applying these approximations to corresponding components of $f_{\sf{Q}}(\rho,\gamma)$, $(\breve{C}6)$ can be convexified as $(\tilde{C}6)$.

\vspace{-1mm}

\section*{Acknowledgment} \vspace{-1mm}
\small
This work was funded by the Luxembourg National Research Fund (FNR), with granted SENTRY project corresponding to grant reference C23/IS/18073708/SENTRY.
 \normalsize

%\vspace{-1mm}
\bibliographystyle{IEEEtran}
%\balance
\bibliography{Journal}
\end{document}